\documentclass[journal]{IEEEtran}

\usepackage{amsmath,amssymb,amsfonts}
\usepackage{graphicx}
\usepackage{cite}
\usepackage{amsmath,amssymb,amsfonts}
\usepackage{amsthm}                          

\newtheorem{theorem}{Theorem}                
\newtheorem{corollary}{Corollary}    
\usepackage{color}
\usepackage{bm}
\usepackage{algorithm}
\usepackage{algorithmic}
\usepackage{url}
\usepackage{balance}

\newcommand{\E}{\mathbb{E}}
\newcommand{\R}{\mathbb{R}}

\newcommand{\jj}{\mathrm{j}}                      
\newcommand{\T}{\mathrm{T}}                        
\newcommand{\CRB}{\mathrm{CRB}}
\newcommand{\FIM}{\mathbf{J}}
\newcommand{\param}{\bm{\theta}}
\newcommand{\sm}{\sigma_{\mathrm{m}}}              

\begin{document}

\title{Cram\'{e}r--Rao Bounds for Magneto-Inductive Integrated Sensing and Communications}

\author{Haofan~Dong,~\IEEEmembership{Student Member,~IEEE,}
        and~Ozgur~B.~Akan,~\IEEEmembership{Fellow,~IEEE}
\thanks{H.~Dong and O.~B.~Akan are with the Internet of Everything Group,
Department of Engineering, University of Cambridge, Cambridge CB3~0FA, U.K.
(e-mail: hd489@cam.ac.uk; oba21@cam.ac.uk).}
\thanks{Ozgur B. Akan is also with the Center for neXt-generation Communications
(CXC), Department of Electrical and Electronics Engineering, Koç University, 34450 Istanbul, Turkey (email:oba21@cam.ac.uk)}}

\maketitle

\begin{abstract}
Magnetic induction (MI) enables communication in RF-denied environments (underground, underwater, in-body), where the medium conductivity imprints a deterministic signature on the channel. This letter derives a closed-form Cram\'{e}r--Rao bound (CRB) for the joint estimation of range and medium conductivity from MI pilot observations in an integrated sensing and communication (ISAC) framework. The Fisher information matrix reveals that the joint estimation penalty converges to 3\,dB in the near-field regime, meaning conductivity sensing adds at most a factor-of-two loss in ranging precision. Monte Carlo maximum-likelihood simulations confirm that the CRB is achievable under practical operating conditions.
\end{abstract}


\section{Introduction}\label{sec:intro}

Integrated sensing and communication (ISAC) has attracted
considerable attention as an enabling technology for
sixth-generation (6G) wireless systems, allowing a single waveform to
serve both data transmission and environmental
perception~\cite{Liu2022JSAC, Zhang2021JSTSP}.
Current ISAC research, however, predominantly
focuses on radio-frequency (RF) links operating at millimeter-
wave or sub-terahertz bands, where the sensing function relies
on target echoes to extract range, velocity, and angle~\cite{Liu2022JSAC,Liu2022COMST}. These
RF-centric designs implicitly assume free-space or near-free-
space propagation and therefore break down in lossy, RF-
denied media such as soil, seawater, and biological tissue,
where electromagnetic waves suffer prohibitive attenuation.

Magnetic induction (MI) communication offers a distinct
physical layer for such environments. Unlike far-field RF
propagation, MI exploits near-field inductive coupling between
coils, with a quasi-static magnetic field that is largely immune
to multipath fading and experiences negligible change in
medium permeability~\cite{Sun2010,Ma2025COMST}. These properties have made
MI a candidate technology for wireless underground
sensor networks (WUSNs)~\cite{Saeed2019COMST,Kisseleff2018}. 
A recent framework termed magneto-inductive ISAC (MI-ISAC) 
has shown that the deterministic MI channel structure enables 
joint communication and sensing in RF-denied 
environments~\cite{dong2026miisac}. Notably, the MI channel 
response carries a bijective imprint of the surrounding medium 
conductivity $\sigma_{\mathrm{m}}$ through eddy-current-induced 
attenuation and phase rotation~\cite{Sun2010}.
This physics offers
a natural opportunity for opportunistic sensing~\cite{Messer2006Science}:
just as commercial microwave links have been repurposed to
monitor rainfall from signal attenuation, an MI communication
link can simultaneously sense the medium it traverses without
dedicated sensing hardware. The sensed conductivity is of
direct practical value: in precision agriculture, soil
conductivity correlates with moisture content and
salinity~\cite{Saeed2019COMST}; in pipeline monitoring,
conductivity changes indicate fluid composition or
corrosion onset~\cite{Kisseleff2018}.

Despite this potential, existing MI literature treats medium
conductivity exclusively as a known system parameter. Channel 
models~\cite{Sun2010,Ma2025COMST} and network 
optimization~\cite{Kisseleff2018} assume $\sigma_{\mathrm{m}}$ is
given \emph{a priori}, while MI localization studies derive Cram\'{e}r--Rao 
bounds (CRBs) for position coordinates with $\sigma_{\mathrm{m}}$ held
fixed~\cite{Saeed2019}. While~\cite{dong2026miisac} 
established the MI-ISAC framework for geometric parameter estimation, 
it did not address the joint range--conductivity problem nor 
quantify the estimation penalty when both must be resolved 
simultaneously. The recent comprehensive survey by Ma 
\emph{et al.}~\cite{Ma2025COMST} explicitly identifies integrated MI
communication, navigation, and sensing as an open research
direction, further motivating this investigation.

In this letter, we extend the MI-ISAC framework 
to joint range--conductivity estimation and develop its 
estimation-theoretic analysis. The main contributions are:
\begin{enumerate}
\item A closed-form Fisher information matrix (FIM) is derived
for joint range and conductivity estimation from MI pilots,
revealing how geometric spreading and eddy-current absorption
contribute separable information.
\item An analytical penalty formula shows the CRB increase
due to joint estimation converges to 3\,dB in the
near-field regime, i.e., at most a factor-of-two loss
in ranging precision.
\item Monte Carlo MLE simulations ($5000$ trials) validate
the CRB achievability.
\end{enumerate}
%

\section{System Model}\label{sec:model}

\subsection{MI-ISAC Scenario}\label{sec:scenario}

We consider a magnetic induction (MI) link between two coaxial
coil transceivers embedded in a lossy medium (Fig.~\ref{fig:system}).
Each coil has radius~$a$ and $N$~turns; the coils are separated by
distance~$r$ in a medium of unknown conductivity~$\sm$.
The transmitter sends known pilots that serve both communication
and channel estimation~\cite{Sun2010, Kisseleff2018}.

\begin{figure}[!t]
  \centering
  \includegraphics[width=0.5\textwidth]{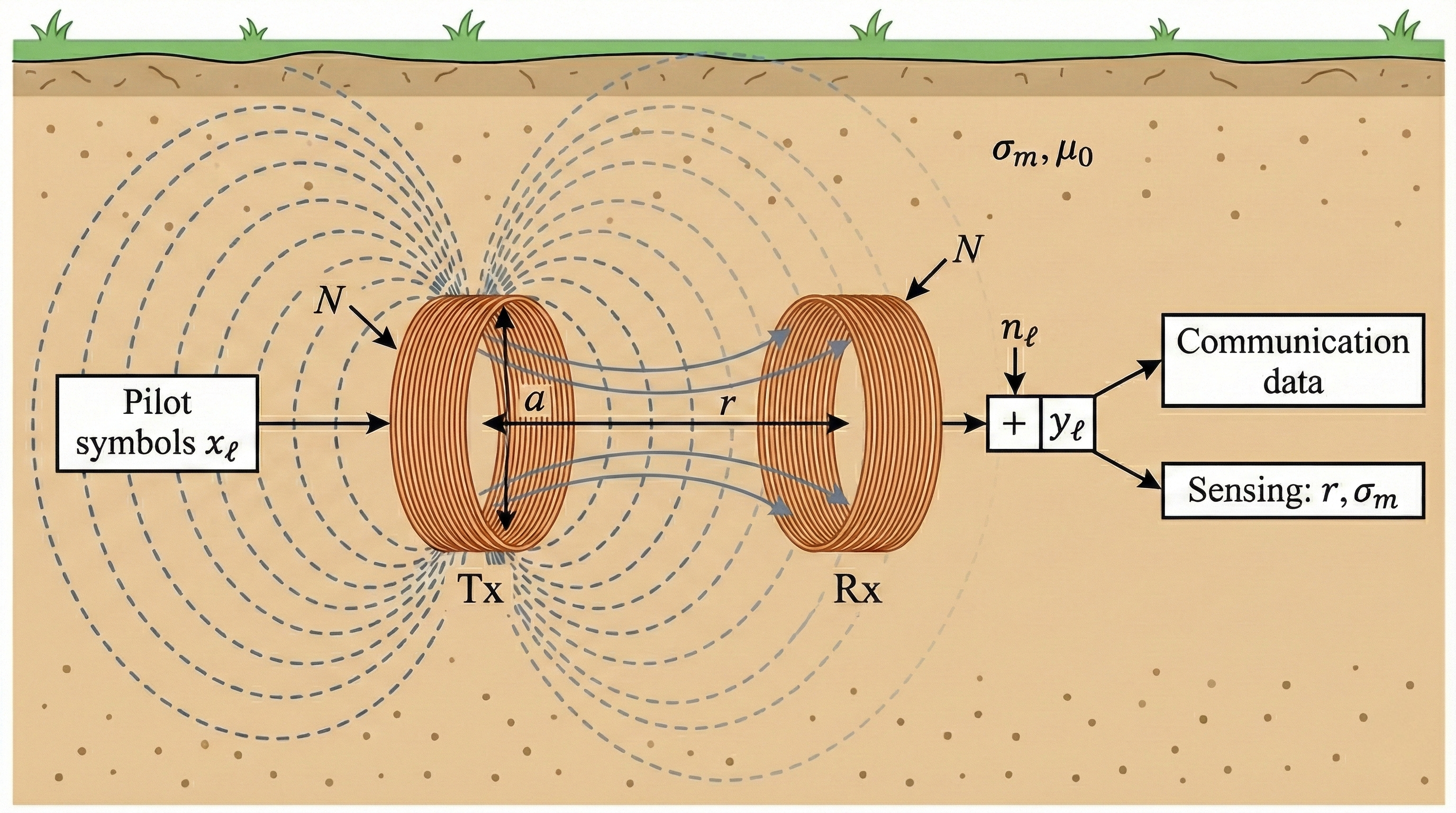}
  \caption{MI-ISAC system model.  Two coaxial coils (radius~$a$,
  $N$~turns) are separated by range~$r$ in a lossy medium
  (conductivity~$\sm$, permeability~$\mu_0$).  The transmitter
  sends $L$~pilot symbols; the receiver jointly estimates
  $\param = [r,\, \sm]^{\T}$ from the received signal, while also
  decoding communication data.}
  \label{fig:system}
\end{figure}

\subsection{MI Channel Model}\label{sec:channel}

The mutual inductance between two coaxial circular coils in a
homogeneous lossy medium with conductivity~$\sm$ and
permeability~$\mu_0$ is~\cite{Sun2010, Akyildiz2015}
\begin{equation}\label{eq:mutual}
  M(r,\sm) = \frac{\mu_0 \pi N^2 a^4}{2\, r^3}\,
  e^{-(1+\jj)\,\alpha\, r},
\end{equation}
where
\begin{equation}\label{eq:alpha}
  \alpha(\sm) = \sqrt{\pi f_0 \mu_0 \sm}
\end{equation}
is the field attenuation constant at operating frequency~$f_0$,
and the skin depth is $\delta = 1/\alpha$.
The exponential factor in~\eqref{eq:mutual} captures the
eddy-current loss: the real part of the exponent models amplitude
decay, while the imaginary part introduces a phase shift that
grows with the electrical distance~$\alpha r$.

It is convenient to introduce the dimensionless parameter
\begin{equation}\label{eq:kappa_r}
  \kappa r \triangleq \sqrt{2}\,\alpha\, r
    = r\sqrt{2\pi f_0 \mu_0 \sm}\,,
\end{equation}
which quantifies the ratio of coil separation to skin depth.
The regime~$\kappa r \ll 1$ corresponds to the \emph{near field}
where medium loss is negligible, while $\kappa r \gg 1$ indicates
heavy attenuation.

The normalized complex channel gain observed at the receiver is
\begin{equation}\label{eq:channel}
  h(r,\sm) = \frac{\omega\, M(r,\sm)}{Z_{\mathrm{ref}}}
  = \frac{C}{r^3}\, e^{-(1+\jj)\,\alpha\, r},
\end{equation}
where $\omega = 2\pi f_0$, $Z_{\mathrm{ref}}$ is a reference
impedance that normalizes the voltage transfer, and
$C \triangleq \omega \mu_0 \pi N^2 a^4 / (2\, Z_{\mathrm{ref}})$
absorbs all frequency- and geometry-dependent constants.
Note that $h$~depends on the two unknowns~$(r, \sm)$
through the near-field decay~$r^{-3}$ and the exponential
medium loss~$e^{-(1+\jj)\alpha r}$, respectively.

\subsection{Signal Model}\label{sec:signal}

The transmitter sends $L$~pilot symbols $\{x_\ell\}_{\ell=1}^{L}$
with per-symbol power $P_{\mathrm{tx}} = \E[|x_\ell|^2]$.
The received signal at the $\ell$-th pilot is
\begin{equation}\label{eq:signal}
  y_\ell = h(r,\sm)\, x_\ell + n_\ell,
  \quad \ell = 1,\ldots,L,
\end{equation}
where $n_\ell \sim \mathcal{CN}(0, N_0 B)$ is additive white
Gaussian noise with power spectral density~$N_0$ and
bandwidth~$B$. The channel~$h$ is treated as constant over the
$L$~pilot symbols (quasi-static assumption), which is well
justified for MI links where the near-field coupling changes
only with the physical displacement of the coils~\cite{Sun2010}.

Since the pilot waveform is known, the matched-filter
sufficient statistic for~$h$ is~\cite{Kay1993}
\begin{equation}\label{eq:hhat}
  \hat{h}_{\mathrm{MF}}
  = \frac{\sum_{\ell=1}^{L} y_\ell\, x_\ell^{*}}
         {\sum_{\ell=1}^{L} |x_\ell|^2}
  = h(r,\sm) + \tilde{n},
\end{equation}
where $\tilde{n} \sim \mathcal{CN}(0, N_0 B / (L\, P_{\mathrm{tx}}))$.
The observation~$\hat{h}_{\mathrm{MF}}$ contains all
information about the parameter vector
\begin{equation}\label{eq:theta}
  \param = \begin{bmatrix} r \\ \sm \end{bmatrix} \in \R^{2}_{>0}
\end{equation}
that is available in the pilot data.

\subsection{Problem Formulation}\label{sec:problem}

The goal is to characterize the fundamental estimation accuracy for
the \emph{joint} recovery of~$\param$ from~\eqref{eq:hhat} via
the Cram\'{e}r--Rao bound (CRB).
Specifically, for any unbiased estimator~$\hat{\param}$,
$\E[(\hat{\param}-\param)(\hat{\param}-\param)^{\T}] \succeq
\FIM^{-1}(\param)$, where~$\FIM(\param)$ is the $2\times 2$
Fisher information matrix (FIM).
The \emph{joint estimation penalty} is defined as
\begin{equation}\label{eq:penalty_def}
  \Delta_r \triangleq
  \frac{\CRB_r^{\mathrm{joint}}}{\CRB_r^{\mathrm{single}}}
  = \frac{[\FIM^{-1}]_{11}}{1/J_{11}} \geq 1,
\end{equation}
where $\CRB_r^{\mathrm{single}} = 1/J_{11}$ is the bound when
$\sm$ is known.  The penalty~$\Delta_r$ quantifies how much
harder it is to estimate range when the medium conductivity
must be simultaneously inferred.
The penalty for~$\sm$ is defined analogously.
Note that $\hat{h}_{\mathrm{MF}}$ provides two real degrees of
freedom (amplitude and phase), which suffices for the two
real unknowns~$(r,\sm)$; identifiability follows from the
injectivity of~$h(r,\sm)$ on~$\R^{2}_{>0}$, since the
near-field decay~$r^{-3}$ and the eddy-current phase
$e^{-\jj \alpha r}$ impose linearly independent constraints.

\section{Cram\'{e}r--Rao Bound Analysis}\label{sec:crb}
\begin{figure*}[!t]
  \centering
  \includegraphics[width=\linewidth]{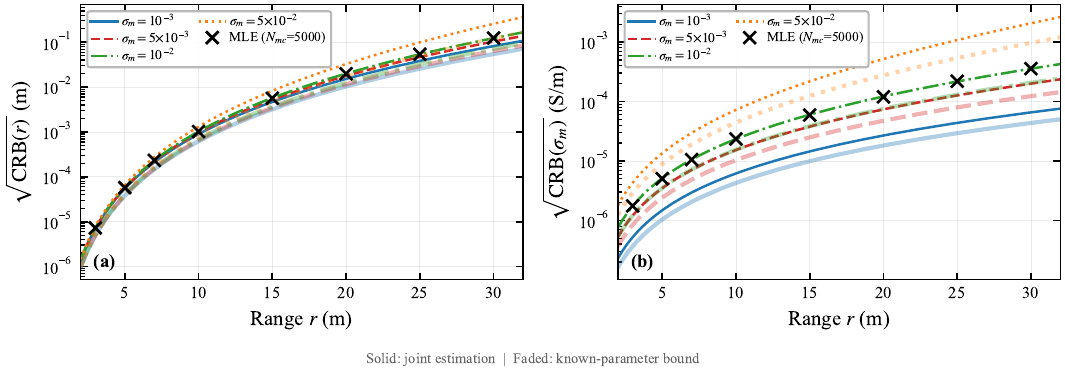}
  \caption{Joint vs.\ single-parameter CRB for (a)~range and
  (b)~conductivity estimation across four medium types.
  Solid: joint estimation; faded: known-parameter bound.
  Cross markers: MLE with $N_{\mathrm{mc}} = 5000$.}
  \label{fig:joint_vs_single}
\end{figure*}

This section derives the Fisher information matrix for the joint estimation of
$\param = [r,\, \sm]^{\T}$ from the sufficient statistic
$\hat{h}_{\mathrm{MF}}$ in~\eqref{eq:hhat}, and characterizes
the penalty incurred by estimating both parameters simultaneously.

Since $\hat{h}_{\mathrm{MF}} \sim \mathcal{CN}\!\big(h(\param),\,
N_0 B / (L P_{\mathrm{tx}})\big)$, the $(i,j)$-th entry of the
$2 \times 2$ FIM for the deterministic (non-random) parameter
vector~$\param$ is~\cite{Kay1993}
\begin{equation}\label{eq:fim_general}
  J_{ij}
  = \frac{2\,L\,P_{\mathrm{tx}}}{N_0 B}\;
    \mathrm{Re}\!\left[
    \frac{\partial h}{\partial \theta_i}\,
    \frac{\partial h^{*}}{\partial \theta_j}
    \right].
\end{equation}
The required partial derivatives of $h$ in~\eqref{eq:channel}
are
\begin{align}
  \frac{\partial h}{\partial r}
  &= h(r,\sm)\,\big[\!-\beta - \jj\,\alpha\big],
  \label{eq:dh_dr}
  \\[3pt]
  \frac{\partial h}{\partial \sm}
  &= h(r,\sm)\,\bigg[\!-(1+\jj)\,\frac{\alpha r}{2\sm}\bigg],
  \label{eq:dh_ds}
\end{align}
where $\beta \triangleq 3/r + \alpha$ and we have used
$\partial\alpha / \partial\sm = \alpha / (2\sm)$ from~\eqref{eq:alpha}.
The derivative~\eqref{eq:dh_dr} contains a real part~$\beta$
arising from the near-field decay~$r^{-3}$ and the attenuation
$e^{-\alpha r}$, and an imaginary part~$\alpha$ from the
medium-induced phase rotation.  The derivative~\eqref{eq:dh_ds}
is proportional to~$(1+\jj)$, reflecting the equal sensitivity
of the channel amplitude and phase to changes in~$\sm$.

Substituting~\eqref{eq:dh_dr}--\eqref{eq:dh_ds}
into~\eqref{eq:fim_general} and factoring
$\gamma \triangleq 2LP_{\mathrm{tx}}|h|^2 / (N_0 B)$, the FIM
takes the form
\begin{equation}\label{eq:fim}
  \FIM = \gamma
  \begin{bmatrix}
    \beta^2 + \alpha^2
    &\; \dfrac{\alpha r}{2\sm}\,(\beta + \alpha)
    \\[6pt]
    \dfrac{\alpha r}{2\sm}\,(\beta + \alpha)
    &\; \dfrac{\alpha^2 r^2}{2\sm^2}
  \end{bmatrix}.
\end{equation}
The off-diagonal term is nonzero whenever
$\alpha > 0$ (i.e., $\sm > 0$), indicating that range and
conductivity estimation are inherently coupled through the
medium loss.
The following result characterizes this coupling.

\begin{theorem}\label{thm:penalty}
Let $\rho$ denote the FIM correlation coefficient
$\rho \triangleq J_{12}/\sqrt{J_{11}\,J_{22}}$.
Then
\begin{equation}\label{eq:rho}
  \rho(\alpha r)
  = \frac{3 + 2\alpha r}
         {\sqrt{2\big[(3 + \alpha r)^2 + (\alpha r)^2\big]}}\,,
\end{equation}
and the joint estimation penalties for~$r$ and~$\sm$ are
\begin{equation}\label{eq:penalty}
  \Delta_r = \Delta_{\sm}
  = \frac{1}{1 - \rho^2} \geq 1.
\end{equation}
\end{theorem}

\begin{IEEEproof}
The diagonal and off-diagonal entries of~$\FIM$ in~\eqref{eq:fim}
yield $J_{11} = \gamma(\beta^2+\alpha^2)$,
$J_{22} = \gamma\alpha^2 r^2/(2\sm^2)$, and
$J_{12} = \gamma \alpha r (\beta+\alpha)/(2\sm)$.
Direct computation gives
$J_{12}^2/(J_{11}\,J_{22})$
$= (\beta+\alpha)^2 / [2(\beta^2+\alpha^2)]$,
which equals~$\rho^2$ in~\eqref{eq:rho} after substituting
$\beta = 3/r + \alpha$.
For any $2 \times 2$ positive-definite matrix, the inverse
satisfies $[\FIM^{-1}]_{ii} = J_{jj}/\det(\FIM)$, so
$\Delta_r = J_{11}[\FIM^{-1}]_{11}
= J_{11}J_{22}/\det(\FIM) = 1/(1-\rho^2)$.
The same identity holds for~$\Delta_{\sm}$ by symmetry.
\end{IEEEproof}

\begin{corollary}[Near-field limit]\label{cor:3dB}
As $\alpha r \to 0$ (equivalently $\kappa r \to 0$),
the correlation coefficient satisfies
$\rho \to 1/\sqrt{2}$,
and hence
\begin{equation}\label{eq:3dB}
  \lim_{\kappa r \to 0}\, \Delta_r
  = \lim_{\kappa r \to 0}\, \Delta_{\sm} = 2
  \quad (\text{i.e., } 3.01~\text{dB}).
\end{equation}
\end{corollary}
This follows by setting $\alpha r = 0$ in~\eqref{eq:rho},
giving $\rho = 3/\sqrt{18} = 1/\sqrt{2}$.

The 3\,dB penalty has a geometric origin: as $\kappa r\!\to\!0$,
$\partial h/\partial r$ is dominated by the real-valued
decay~$r^{-3}$ while $\partial h/\partial \sm$ retains its
$(1+\jj)$ structure, fixing their angle in the complex plane at
$\rho = 1/\sqrt{2}$ independent of system parameters.

\emph{Remark:}
In the opposite limit $\alpha r \to \infty$,
\eqref{eq:rho} gives $\rho \to 1$, and the penalty diverges.
Physically, both partial derivatives become proportional to
$(1+\jj)\,h\,e^{-(1+\jj)\alpha r}$, so their ``directions'' in the
complex plane align and the FIM becomes rank-deficient;
$r$ and~$\sm$ are then asymptotically indistinguishable from a
single observation.
This motivates operating in the regime $\kappa r < 1$, where
the penalty remains below approximately 5\,dB (see Section~\ref{sec:results}).

\section{Numerical Results}\label{sec:results}

\begin{figure}[!t]
  \centering
  \includegraphics[width=0.45\textwidth]{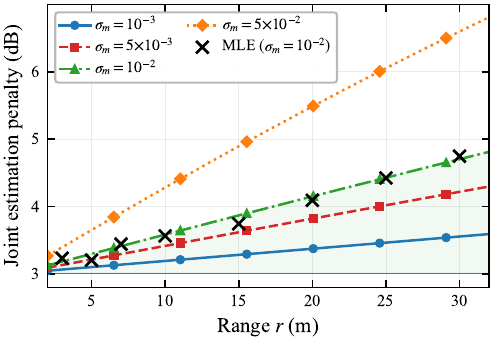}
  \caption{Joint estimation penalty vs.\ range for four medium
  conductivities.  Horizontal line: 3\,dB near-field limit.
  Cross markers: MLE empirical penalty ($N_{\mathrm{mc}} = 5000$,
  $\sm = 0.01$\,S/m).}
  \label{fig:penalty_sweep}
\end{figure}

\begin{figure*}[!t]
  \centering
  \includegraphics[width=1\textwidth]{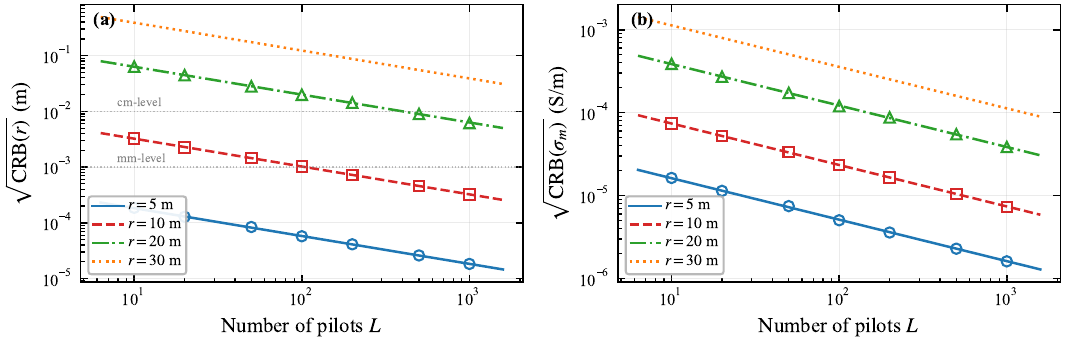}
  \caption{CRB vs.\ pilot count $L$ for (a)~range and
  (b)~conductivity estimation at multiple distances
  ($\sm = 0.01$\,S/m).  Open markers: MLE validation.
  Dashed lines: cm- and mm-level accuracy thresholds.}
  \label{fig:design}
\end{figure*}

The analytical CRB and penalty results are validated through
Monte Carlo simulations with maximum likelihood estimation (MLE).
Unless stated otherwise, the default parameters are:
operating frequency $f_0 = 10$\,kHz, coil radius $a = 0.15$\,m,
number of turns $N = 20$, pilot count $L = 100$,
transmit power $P_{\mathrm{tx}} = 0$\,dBm, and noise bandwidth
$B = 1$\,kHz at temperature $T_0 = 290$\,K.
The MLE is implemented via multi-start Nelder--Mead optimization
of the negative log-likelihood
$\mathcal{L}(\param) = L\,|\hat{h}_{\mathrm{MF}} - h(\param)|^2 / (N_0 B)$,
with $N_{\mathrm{mc}} = 5000$ independent trials per operating point.

Fig.~\ref{fig:joint_vs_single} compares the joint CRB (solid)
with the known-parameter bound (faded) for four conductivities.
The MLE variance (crosses, $\sm = 0.01$\,S/m) tracks the joint
CRB, and a consistent $\approx$3\,dB gap is visible across all
media, confirming Theorem~\ref{thm:penalty}.

Fig.~\ref{fig:penalty_sweep} shows the penalty versus range.
All curves converge to the 3\,dB floor at short range; higher
conductivity causes faster growth due to larger~$\kappa r$.
For typical underground conditions ($\sm \in [10^{-3},
10^{-1}]$\,S/m, $r < 20$\,m), the penalty stays below 5\,dB.

Fig.~\ref{fig:design} illustrates the system design space
by plotting $\sqrt{\CRB}$ against the pilot count~$L$ for
multiple ranges.  Both range and conductivity estimation
accuracy improve as $1/\sqrt{L}$, consistent with the
$L$-scaling in~\eqref{eq:fim_general}.  At $r = 10$\,m,
centimeter-level range accuracy is achievable with
$L \geq 50$ pilots, while $r = 20$\,m requires $L \geq 200$.
The MLE markers (open symbols) match the analytical curves
across the full sweep, validating the CRB achievability
for all tested configurations.

\section{Conclusion}\label{sec:conclusion}

This letter has derived closed-form Cram\'{e}r--Rao bounds for
joint range and medium conductivity estimation in
magneto-inductive ISAC systems.
The analysis reveals a fundamental 3\,dB joint estimation
penalty in the near-field regime, showing that
range and conductivity remain distinguishable owing to
the distinct information contributions of geometric spreading
and eddy-current absorption.
Monte Carlo MLE simulations confirm the tightness of the bound
across the tested parameter space.
Future work will extend the analysis to multi-coil arrays,
frequency-diverse pilots, and experimental validation in
field-deployed underground sensor networks.

\bibliographystyle{IEEEtran}
\bibliography{refs}

@article{Liu2022JSAC,
  author    = {Fan Liu and Yuanhao Cui and Christos Masouros and Jie Xu and Tony Xiao Han and Yonina C. Eldar and Stefano Buzzi},
  title     = {Integrated Sensing and Communications: Toward Dual-Functional Wireless Networks for {6G} and Beyond},
  journal   = {IEEE J. Sel. Areas Commun.},
  volume    = {40},
  number    = {6},
  pages     = {1728--1767},
  month     = jun,
  year      = {2022},
  doi       = {10.1109/JSAC.2022.3156632},
}

@article{Zhang2021JSTSP,
  author    = {J. Andrew Zhang and Md Lushanur Rahman and Kai Wu and Xiaojing Huang and Y. Jay Guo and Shanzhi Chen and Jinhong Yuan},
  title     = {An Overview of Signal Processing Techniques for Joint Communication and Radar Sensing},
  journal   = {IEEE J. Sel. Topics Signal Process.},
  volume    = {15},
  number    = {6},
  pages     = {1295--1315},
  month     = nov,
  year      = {2021},
  doi       = {10.1109/JSTSP.2021.3113120},
}

@article{Liu2022COMST,
  author    = {An Liu and Zhe Huang and Min Li and Yubo Wan and Wenrui Li and Tony Xiao Han and Chenchen Liu and Rui Du and Danny Kai Pin Tan and Jianmin Lu and Yuan Shen and Fabiola Colone and Kevin Chetty},
  title     = {A Survey on Fundamental Limits of Integrated Sensing and Communication},
  journal   = {IEEE Commun. Surveys Tuts.},
  volume    = {24},
  number    = {2},
  pages     = {994--1034},
  year      = {2022},
  doi       = {10.1109/COMST.2022.3149272},
}

@article{Sun2010,
  author    = {Zhi Sun and Ian F. Akyildiz},
  title     = {Magnetic Induction Communications for Wireless Underground Sensor Networks},
  journal   = {IEEE Trans. Antennas Propag.},
  volume    = {58},
  number    = {7},
  pages     = {2426--2435},
  month     = jul,
  year      = {2010},
  doi       = {10.1109/TAP.2010.2048858},
}

@article{Ma2025COMST,
  author    = {Honglei Ma and Erwu Liu and Wei Ni and Zhijun Fang and Rui Wang and Yongbin Gao and Dusit Niyato and Ekram Hossain},
  title     = {Through-the-{Earth} Magnetic Induction Communication and Networking: A Comprehensive Survey},
  journal   = {IEEE Commun. Surveys Tuts.},
  volume    = {28},
  pages     = {2263--2305},
  year      = {2025},
  doi       = {10.1109/COMST.2025.3623258},
}

@article{Kisseleff2018,
  author    = {Steven Kisseleff and Ian F. Akyildiz and Wolfgang H. Gerstacker},
  title     = {Survey on Advances in Magnetic Induction-Based Wireless Underground Sensor Networks},
  journal   = {IEEE Internet Things J.},
  volume    = {5},
  number    = {6},
  pages     = {4843--4856},
  month     = dec,
  year      = {2018},
  doi       = {10.1109/JIOT.2018.2870289},
}

@article{Saeed2019COMST,
  author    = {Nasir Saeed and Mohamed-Slim Alouini and Tareq Y. Al-Naffouri},
  title     = {Toward the {Internet} of Underground Things: A Systematic Survey},
  journal   = {IEEE Commun. Surveys Tuts.},
  volume    = {21},
  number    = {4},
  pages     = {3443--3466},
  year      = {2019},
  doi       = {10.1109/COMST.2019.2934365},
}

@article{Saeed2019,
  author    = {Nasir Saeed and Mohamed-Slim Alouini and Tareq Y. Al-Naffouri},
  title     = {{3D} Localization for {Internet} of Underground Things in Oil and Gas Reservoirs},
  journal   = {IEEE Access},
  volume    = {7},
  pages     = {121769--121780},
  year      = {2019},
  doi       = {10.1109/ACCESS.2019.2937915},
}

@article{Messer2006Science,
  author    = {Hagit Messer and Artem Zinevich and Pinhas Alpert},
  title     = {Environmental Monitoring by Wireless Communication Networks},
  journal   = {Science},
  volume    = {312},
  number    = {5774},
  pages     = {713},
  month     = may,
  year      = {2006},
  doi       = {10.1126/science.1120034},
}

@article{Akyildiz2015,
  author    = {Ian F. Akyildiz and Pu Wang and Zhi Sun},
  title     = {Realizing Underwater Communication Through Magnetic Induction},
  journal   = {IEEE Commun. Mag.},
  volume    = {53},
  number    = {11},
  pages     = {42--48},
  month     = nov,
  year      = {2015},
  doi       = {10.1109/MCOM.2015.7321970},
}

@book{Kay1993,
  author    = {Steven M. Kay},
  title     = {Fundamentals of Statistical Signal Processing: Estimation Theory},
  publisher = {Prentice-Hall},
  address   = {Upper Saddle River, NJ, USA},
  year      = {1993},
}

@article{dong2026miisac,
  author    = {Haofan Dong and Ozgur B. Akan},
  title     = {{MI-ISAC}: Magneto-Inductive Integrated Sensing and Communication in the Reactive Near-Field for {RF}-Denied Environments},
  journal   = {arXiv preprint arXiv:2602.07714
        
        
        
        },
  year      = {2026},
  month     = feb,
  eprint    = {2602.07714},
  archivePrefix = {arXiv},
  primaryClass = {eess.SP},
  note      = {Submitted to IEEE Wireless Commun. Lett.},
  doi       = {10.48550/arXiv.2602.07714},
}

\end{document}